\documentclass[11pt]{article}
\usepackage{geometry}                      
\usepackage{graphicx}
\usepackage{amssymb}
\usepackage{natbib}
\usepackage{amsmath}   
\usepackage{bm}

\newtheorem{theorem}{Theorem}[section]
\newtheorem{lemma}[theorem]{Lemma}

\newenvironment{proof}[1][Proof]{\begin{trivlist}
\item[\hskip \labelsep {\bfseries #1}]}{\end{trivlist}}

\newcommand{\qed}{\nobreak \ifvmode \relax \else
      \ifdim\lastskip<1.5em \hskip-\lastskip
      \hskip1.5em plus0em minus0.5em \fi \nobreak
      \vrule height0.75em width0.5em depth0.25em\fi}

\begin{document}

\title{Numerical Approximation of Probability Mass Functions Via the Inverse Discrete Fourier Transform}
\author{
Richard L. Warr\footnote{Department of Mathematics and Statistics, Air Force Institute of Technology, Wright-Patterson Air Force Base, Ohio, USA.  The views expressed in this article are those of the authors and do not reflect the official policy or position of the United States Air Force, Department of Defense, or the U.S. Government.}
}
\maketitle
\date
\begin{abstract}
First passage distributions of semi-Markov processes are of interest in fields such as reliability, survival analysis, and many others.  The problem of finding or computing first passage distributions is, in general, quite challenging.  We take the approach of using characteristic functions (or Fourier transforms) and inverting them, to numerically calculate the first passage distribution.  Numerical inversion of characteristic functions can be numerically unstable for a general probability measure, however, we show for lattice distributions they can be quickly calculated using the inverse discrete Fourier transform.  Using the fast Fourier transform algorithm these computations can be extremely fast.  In addition to the speed of this approach, we are able to prove a few useful bounds for the numerical inversion error of the characteristic functions.  These error bounds rely on the existence of a first or second moment of the distribution, or on an eventual monotonicity condition.  We demonstrate these techniques in an example and include R-code.
\end{abstract}

\section{Introduction}
Statistical literature abounds with proofs using characteristic functions (CFs).  Asymptotic results such as the central limit theorem rely heavily on the properties of CFs.  However, when it comes to applied statistics the reverse is true.  In general statisticians seem very uncomfortable using the CF or numeric approximations of it, even though they routinely use numerical approximations and calculations for various other procedures.  This could be partly due to the fact that CFs are complex functions, but this is a deterrent based on fear of the unknown, not difficulty.

This is not the case in many applied sciences.  Engineers make heavy use of the Fourier transform and are quite comfortable using it in practice.  The Fourier transform and the characteristic function are the same function with a change of variables.  To use one over the other is a matter of preference, the curves, as drawn, are identical.  There are some areas in applied statistics where use of CFs are helpful.  One such area is semi-Markov processes (SMPs), specifically finding first passage distributions from one state to another, as in statistical flowgraph models \cite{flowgraphs}.  Using characteristic functions is a natural way to find first passage distributions in SMPs.  

For example consider the graphical representation of a SMP in Figure \ref{fig:1}, where the nodes on the graph represent a state, and the branches represent waiting time distributions before transition to the next state.  Each branch is labeled with a characteristic function $\varphi_{_{ij}}(s)$ and a probability of taking that branch.  The characteristic function of the first passage distribution is:
\begin{equation}
\varphi_{_{13*}}(s) = \frac{(1-p)\varphi_{_{12}}(s)\varphi_{_{23}}(s)}{1-p\varphi_{_{12}}(s)\varphi_{_{21}}(s)},
\label{eq:ftpas}
\end{equation}
the ``*" denotes a first passage.  For details on finding the first passage CF see \cite{flowgraphs} or for a more theoretical development refer to \cite{pyke2}.  Once the CFs for the branches are known, it is trivial to calculate the CF of the first passage from state $1$ to $3$. The challenge is to convert the CF to a distribution.  In most cases this must be done numerically.  

\begin{figure}[ht]
\begin{center}
\includegraphics[width=5in]{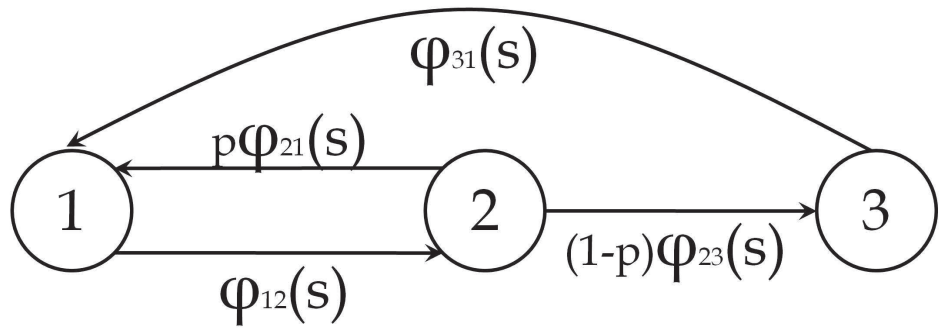}
\caption{An example semi-Markov process.}
\label{fig:1}
\end{center}
\end{figure}

Work has been done to show that the discrete Fourier transform (DFT) can accurately invert characteristic functions given sufficiently fast decay of both the density and the characteristic function (see \cite{FFTBounds}).  One class of distributions that has not been well studied for inversion using the DFT is discrete distributions.  Specifically we focus on distributions with support $n \Delta t$ for $n = 0, \pm1, \pm2, \ldots$ and $\Delta t \in \mathbb{R}^{+}$.  This type of distribution has been termed a \textsl{lattice} or an \textsl{arithmetic} distribution, these include most of the well studied discrete distributions such as the binomial, Poisson, negative binomial and many others.

The primary goal of this paper is to promote the use of the inverse discrete Fourier transform (iDFT) for the inversion of CFs to probability mass functions (PMFs).  A secondary benefit of using the DFT is easy calculation of CFs when they do not have a closed form.

This paper is outlined as follows. The next section details some properties of the DFT, in particular how the DFT is applied when dealing with lattice distributions.  Next, the DFT is shown in action, in an application.  We include code using the R software package (see \cite{R}). Finally, we conclude with a few comments.

\section{The discrete Fourier transform}
\label{sec:1}

For a random variable $T$, with distribution $F$, its characteristic function is
\begin{equation}
\varphi_{_T}(s) \equiv \int_{-\infty}^{\infty} e^{ist} \, dF. 
\end{equation}
Let $\hat{f}(\omega)$ be the Fourier transform of $T$, then the relationship between the characteristic function and the Fourier transform is 
\begin{equation}
\varphi_{_T}(-2\pi\omega) = \hat{f}(\omega).  
\end{equation}
In the reminder of the paper we use Fourier transforms, with the understanding that the CF is easily obtained from the FT.  We do this primarily for ease in applying DFT methods, specifically using the fast Fourier transform (FFT) algorithm.

The DFT can be used to approximate the Fourier transform and its inverse.  Under certain conditions, the error of this approximation can be controlled.  An excellent reference on the DFT is \cite{DFTMan}. To use the DFT, the support of the RV must be discretized at evenly spaced points, where $\Delta t$ is the space between adjacent points.  We term these points the support of the DFT samples.  Similarly the support of the iDFT samples the FT at evenly spaced intervals of width $\Delta \omega$ in the transform domain. These sampling widths are related by the reciprocity relation 
\begin{equation}
\Delta t \Delta \omega = \frac{1}{N}.
\label{eqn:Reciprocity}
\end{equation}
With this information we apply the DFT and iDFT to lattice distributions.
 
\subsection{DFT of lattice distributions}
In this section error bounds for lattice distributions with nonnegative support are derived.  If $f(t_n)$ is the PMF of a lattice distribution, then the DFT of that distribution is:
\begin{equation}
F_k \equiv \sum_{n=0}^{N-1} f(t_n) e^{-2 \pi i \frac{n k}{N}}, \text{ for } k = 0,1, \ldots, N-1.
\label{eqn:DFT}
\end{equation}
Clearly the utility of the DFT as an approximation of the FT is dependent on the size of $N$.  The major benefit of using the DFT with discrete random variables is that unlike RVs with PDFs, when sampling from lattice distributions we are able to capture essentially all the information about the random variable.  In this paper we use the word ``sampling" loosely to mean the PMF values of the support that are included in the DFT.  When sampling a RV with a PDF the measure of our finite number of samples is zero, whereas when sampling from a certain class of lattice distributions our ``sample" has measure that can approach one.  In essence, we are losing little or no information when sampling from a DRV, or in other words, with a finite number of points we obtain more of a ``census" than a ``sample" of where the probability is located in a distribution.  This allows the DFT to accurately approximate the FT at specified points given $N$ sufficiently large.

We confine our focus to nonnegative lattice RVs $T$ where, given any $\varepsilon > 0$, there exists an $N$ such that $\sum_{i=0}^{N-1} P(T=t_i) > 1-\varepsilon$.  Nearly all properly defined nonnegative lattice RVs fall into this class.  There may be a few pathological cases which do not, such as the counting measure on all the integers, however, in many cases one would argue that these are not properly defined distributions.  The assumption that $T$ is nonnegative is only slightly less general than the unrestricted case, and derivations for lattice RVs with unbounded support would be very similar.

To find the error bound for the DFT on lattice distributions one only needs to look at the definition of the Fourier transform.  Let $f(t_n) \equiv P(T = n \, \Delta t)$, then the definition of the Fourier transform is:

\begin{equation}
\hat{f}(\omega) \equiv \sum_{n=0}^{\infty} f(t_n) e^{-2 \pi i \omega n \Delta t } = \sum_{n=0}^{\infty} P(T=n \Delta t) e^{-2 \pi i \omega n \Delta t } 
\label{eqn:FT4DRVs}
\end{equation}

Now our goal is to bound the error $|\hat{f}(\omega_k) - F_k |$ for all index $k=0,1,\ldots,N-1$.  Therefore we must find an optimal $\omega_k$ that matches well with the DFT's $F_k$.  To do this we simply set the exponents of Equations \ref{eqn:DFT} and \ref{eqn:FT4DRVs} equal to each other and solve for $\omega_k$.  Therefore 
\begin{equation}
-2 \pi i \omega_k n \Delta t = -2 \pi i \frac{nk}{N},
\end{equation}
implies that
\begin{equation}
\omega_k = \frac{k}{ \Delta t N}.
\end{equation}
Now using the constraint that 
\[ \sum_{n=0}^{N-1} P(T= n \, \Delta t) = \sum_{n=0}^{N-1} f(t_n) \geq 1- \varepsilon,\]
we show that the error can be controlled.

\begin{lemma}
\label{lemma1}
For a nonnegative lattice random variable, defined on the support $n\, \Delta t$ for $n =0, 1,  \ldots, \infty$, the pointwise error of the forward DFT is less than or equal to $P(T \geq N \, \Delta t)$.
\end{lemma}

\begin{proof}
\begin{eqnarray}
 \left|\hat{f}(\omega_k) - F_k \right| &=& \left|\sum_{n=0}^{\infty} P(T=n \Delta t) e^{-2 \pi i \omega_k n \Delta t} - \sum_{n=0}^{N-1} P(T=n \Delta t) e^{-2 \pi i \frac{n k}{N}} \right|      \nonumber \\
   &=& \left|\sum_{n=0}^{\infty} P(T=n \Delta t) e^{-2 \pi i \frac{n k}{N}} - \sum_{n=0}^{N-1} P(T=n \Delta t) e^{-2 \pi i \frac{n k}{N}} \right|      \nonumber \\
   &=& \left| \sum_{n=N}^{\infty} P(T=n \Delta t) e^{-2 \pi i \frac{n k}{N}} \right|      \nonumber \\
   &\leq&  \sum_{n=N}^{\infty}  P(T=n \Delta t) \left| e^{-2 \pi i \frac{n k}{N}} \right|   \ = \ \sum_{n=N}^{\infty} P(T=n \, \Delta t)  \nonumber \\
   &=& \sum_{n=N}^{\infty} f(t_n) \  =P(T \geq N \Delta t).  \nonumber \qed
\end{eqnarray} 
\end{proof}

To give an example, consider a Poisson random variable $U$, with rate parameter $\lambda=5$.  The FT of $U$ is $\hat{f}(\omega)=\exp\{5e^{-2\pi i\omega}-5\}$.  If we choose $N=16$, where $\Delta t=1$, then $\omega_k=k/16$ for $k = 0,1,\ldots,15$.  We set the error $\varepsilon = P(U > 15) \approx  0.000069$.  The greatest error occurs when computing $|\hat{f}(0/16)-F_{0}| \approx |1-0.999931| =  0.000069$.  Therefore we can see in this example that $\left|\hat{f}(\omega_k) - F_k \right| \leq \varepsilon$.


Now that we have shown that the DFT can accurately approximate the Fourier transform of a DRV, at specified points, we focus on the inverse DFT.  This is the more difficult of the two problems, but is in general more useful.  Given the Fourier transform of a DRV we now demonstrate how to accurately calculate the PMF.  



Given that $\hat{f}(\omega)$ is a known function, we show that the inverse DFT has the same error bound as the forward DFT.  The inverse DFT is defined as:
\begin{equation}
f_k \equiv \frac{1}{N}\sum_{n=0}^{N-1} \hat{f}(\omega_n) e^{2 \pi i \omega_n k \Delta t}. 
\label{eqn:idft}
\end{equation}

\begin{lemma}
\label{lemma2}
Given the Fourier transform of a nonnegative lattice RV, $T$, evaluated at points $k/(N \, \Delta t$), for $k=0, 1,  \ldots N-1$, the pointwise error of the inverse DFT is less than or equal to $P(T \geq N \, \Delta t)$.
\end{lemma}
We introduce a couple of notational items to help with the following proof.  Let ``mod" be the modulo operator, so for the positive integer $b$, and non-negative integers $a$, $n$ and $r$; then $r = a \text{ mod } b$, where $n$ is the largest non-negative integer that satisfies the equation $a=r+nb$.  This implies that $r<b$.  Also, let
\begin{equation}
I_{(\text{condition})} = \left\{ \begin{array}{ll}
1 & \mbox{if ``condition" is true} \\
0 & \mbox{otherwise}
\end{array}
\right.
\end{equation}
We'll call $I_{(\text{condition})}$ the indicator function.
\begin{proof} 
Using Equations \ref{eqn:FT4DRVs} and \ref{eqn:idft} we can write the iDFT as:
\begin{eqnarray}
f_k &=& \frac{1}{N} \sum_{n=0}^{N-1} \sum_{j=0}^{\infty} P(T=j \Delta t) e^{-2 \pi i \omega_n j \Delta t} e^{2 \pi i \omega_n k \Delta t}      \nonumber \\
&=& \frac{1}{N} \sum_{n=0}^{N-1} \sum_{j=0}^{\infty} P(T=j \Delta t) e^{-2 \pi i \frac{nj}{N}} e^{2 \pi i \frac{nk}{N}}      \nonumber \\
&=& \frac{1}{N} \sum_{j=0}^{\infty} P(T=j \Delta t) \sum_{n=0}^{N-1}  e^{-2 \pi i \frac{nj}{N}} e^{2 \pi i \frac{nk}{N}}      \nonumber \\
&=& \frac{1}{N} \sum_{j=0}^{\infty} P(T=j \Delta t) \, N \, \text{I}_{(k= j\text{ mod }N)}      \nonumber \\
&=&  \sum_{j=0}^{N-1} P(T=j \Delta t) \, \text{I}_{(k= j\text{ mod }N)}  +    \sum_{j=N}^{\infty} P(T=j \Delta t) \, \text{I}_{(k= j\text{ mod }N)} \nonumber \\
&=&  P(T=k \Delta t)  +    \sum_{j=N}^{\infty} P(T=j \Delta t) \, \text{I}_{(k= j\text{ mod }N)}. \nonumber 
\end{eqnarray}
Therefore we can now express the error as:
\begin{eqnarray}
\left| f_k - P(T=k \Delta t)  \right| & = & \left| P(T=k \Delta t) + \sum_{j=N}^{\infty} P(T=j \Delta t) \, \text{I}_{(k= j\text{ mod }N)} - P(T=k \Delta t)  \right| \nonumber \\
& \leq & P(T \geq N \Delta t).  \nonumber \qed
\end{eqnarray}
\end{proof}
This result is not surprising, because of the close relationship between the forward and inverse DFT.

The next result assumes the the FT is unknown and must be approximated by the DFT, following which, the iDFT is used to invert the approximated FT to a PMF.  This theorem shows this error is also bounded.

\begin{theorem}
\label{thm1}
Given that the FT of a nonnegative lattice random variable is known up to a fixed error bound, then the error bound of the iDFT of this FT can be bounded by $2P(T \geq N \, \Delta t)$.
\end{theorem}

\begin{proof}
Mimicking the proof of Lemma \ref{lemma2} we replace $\hat{f}(\omega_n)$ with $\hat{f}(\omega_n)+\delta_n$, where $\varepsilon=P(T \geq N \, \Delta t)$ and $|\delta_n|\leq \varepsilon$.  This gives:
\begin{eqnarray}
f_k &=& \frac{1}{N} \sum_{n=0}^{N-1} \left[ \hat{f}(\omega_n)+\delta_n \right] e^{2 \pi i \omega_n k \Delta t}      \nonumber \\
&=& \frac{1}{N} \sum_{n=0}^{N-1}  \hat{f}(\omega_n) e^{2 \pi i \frac{nk}{N}} + \frac{1}{N} \sum_{n=0}^{N-1} \delta_n  e^{2 \pi i \frac{nk}{N}}     \nonumber \\
&=& P(T=k \Delta t) + \sum_{j=N}^{\infty} P(T=j \Delta t) \, \text{I}_{(k= j\text{ mod }N)} + \frac{1}{N} \sum_{n=0}^{N-1} \delta_n  e^{2 \pi i \frac{nk}{N}}      \nonumber  
\end{eqnarray}
Therefore if we only have an approximation for $\hat{f}(\omega_n)$ then,
\begin{eqnarray}
 \left| f_k - P(T=k \Delta t)  \right| & = & \left|  \sum_{j=N}^{\infty} P(T=j \Delta t) \, \text{I}_{(k= j\text{ mod }N)} + \frac{1}{N} \sum_{n=0}^{N-1} \delta_n  e^{2 \pi i \frac{nk}{N}}   \right| \nonumber \\
 &\leq&  \left|  \sum_{j=N}^{\infty} P(T=j \Delta t) \, \text{I}_{(k= j\text{ mod }N)} \right| + \left| \frac{1}{N} \sum_{n=0}^{N-1} \delta_n  e^{2 \pi i \frac{nk}{N}}   \right| \nonumber \\
 &\leq& \varepsilon +  \frac{1}{N} \sum_{n=0}^{N-1} \left| \delta_n \right| \left| e^{2 \pi i \frac{nk}{N}}   \right| \nonumber \\
 &\leq& \varepsilon +  \frac{\varepsilon}{N} \sum_{n=0}^{N-1}  1 \ = \ 2\varepsilon.  \nonumber \qed 
\end{eqnarray}
\end{proof}

What we have not addressed, thus far, is the error introduced when manipulating FTs in the Fourier domain.  The primary reason we use CFs and FTs is for the convenient theoretical properties but when $\hat{f}(\omega_n)$ is estimated, then when mathematical manipulation such as multiplication or division occurs the error is magnified.  We do not specifically address this problem, but warn the practitioner of this added source of error.  Additionally, machine precision error also plays a part in the overall error bound.

\subsection{Finding the inversion error bounds}

The primary problem when using the inverse DFT is that one may not know \textit{a priori} $N$ such that
\[\sum_{n=0}^{N-1} P(T= n \Delta t) \geq 1- \varepsilon,\]
for a given $\varepsilon$.  So although the error is bounded, this bound is not obvious without the exact PMF, which we do not usually have.  There are a few approaches to overcome this difficulty.  One solution is to use an inequality to bound $P(T> n\Delta t)$.  One bound is Markov's inequality:
\[ P(T \geq a) \leq \frac{E(T)}{a}, \ \ \text{ for } a > 0.  \]
Therefore if the first moment exists for the RV $T$, and we either know or can estimate it, this method provides a conservative error bound.  The primary issue is finding $E(T)$.  In our case the random variable $T$'s distribution is a first passage distribution in a semi-Markov process.  It turns out if all the expectations to the individual transitions in the SMP are known, then finding the first passage expectation is the solution to a set of linear equations. Literature exists that explains how one can find the moments of first passage distributions, for example see \cite{SMPFirstPass} and \cite{MPFirstPass}.  

Another inequality that can assist in bounding the error is Cantelli's inequality.  For a random variable $T$ with $\mu=E(T)$, $\sigma^2=Var(T)$, and $k>0$ we have
\[ P(T \geq k\sigma +\mu) \leq \frac{1}{1+k^2}.  \]
Therefore if $Var(T) < \infty$, and given $E(T)$ and $Var(T)$, Cantelli's inequality can be useful in bounding the DFT error.  Where $k > \sigma/\mu$, Cantelli's inequality provides a smaller bound and should be used over Markov's inequality.  Finding the variance of a first passage random variable is similar to the expectation, again see \cite{MPFirstPass} for details.

If the two inequalities prove to be impractical, then another approach can be taken if certain assumptions can be made about the distribution of $T$.  
\begin{theorem}
\label{thm2}
If for some index $M$ where the following inequalities are satisfied $P(T=n \Delta t) \leq P(T=m \Delta t)$ and $M<m<n$, then the error for the iDFT can be expressed as:
Choose an even $N$ (the number of DFT samples) such that $N > 2M$, then for all $n= 0,1, \ldots, (N/2-1)$
\[ \left| f_n - P(T=n \Delta t)  \right| \leq f_{N/2+n}. \]
\end{theorem}
Essentially what this is saying is that after a certain point, if the first passage PMF, $f(t)$, is monotonically nonincreasing, then we can bound the error using that fact alone.  This technique is certainly not universally applicable, but it does provide a working error bound if monotonicity eventually occurs.  To our knowledge PMFs used for modeling processes will often meet this assumption.

\begin{proof}
From the proof of Lemma \ref{lemma2} we know
\[ f_n = \sum_{j=0}^{\infty} P(T=j \Delta t) \, \text{I}_{(n= j\text{ mod }N)}, \]
Which can be rewritten as
\[  f_n = \sum_{j=0}^{\infty} P(T=(n+jN) \Delta t).   \]
Which implies
\[  f_{N/2+n} = \sum_{j=0}^{\infty} P(T=(N/2+n+jN) \Delta t)   \]
therefore 
\[ \left| f_n - P(T=n \Delta t)  \right| = \sum_{j=1}^{\infty} P(T=(n+jN) \Delta t)  \]
but by monotonicity we know
\[  P\left(T=\left(N/2+n+jN\right) \Delta t \right) \geq P\left(T=\left(n+(j+1)N\right) \Delta t \right)  \text{ for all } j = 0, 1, 2 \ldots  \]
This implies 
\[  f_{N/2+n} \geq  \sum_{j=1}^{\infty} P(T=(n+jN) \Delta t) .   \]
Therefore 
\[ \left| f_n - P(T=n \Delta t)  \right| \leq f_{N/2+n}. \qed \]
\end{proof}

This bound can be the more effective then the previous bounds, if the condition exists for it to be applied.  However, it is difficult to prove the monotonicity condition in a particular situation, and if sufficient doubt exists, another error bounding method should be applied.  

Before we move to the application section, we suggest an even broader condition than eventual monotonicity that will provide the same error bound.  This condition is: 
\begin{theorem}
\label{thm3}
If there exists some nonnegative integers $n$ and $N/2$ and a positive integer $k$, such that, if $n<N/2$, and $P\left(T=\left(n+kN/2\right) \Delta t \right) \geq P\left(T=\left(n+(k+1)N/2\right) \Delta t \right)$ then, for a DFT with $N$ samples,
\[ \left| f_n - P(T=n \Delta t)  \right| \leq f_{n+N/2}. \]
\end{theorem}
These types of distributions might be termed \textit{eventually periodic nonincreasing}.  This condition is very broad and can be applied in nearly all first passage distributions, although we are sure there exist some pathological cases that do not fit in this class of distributions.  The proof of Theorem \ref{thm3} follows the same argument as Theorem \ref{thm2}.

In the next section we demonstrate using the DFT to calculate FTs, their inversion, and also it's error bounds.

\section{An application of discrete time semi-Markov processes}
Consider the problem described in \cite{barbu}, the waste treatment for a textile factory.  If the treatment facility is working then the waste can be discarded; if the unit fails the untreated waste must be stored in a holding tank until the treatment facility is again functioning.  However, if the repair to the treatment facility takes too long the holding tank becomes full, and the textile factory must halt production until the repair is complete.  Figure \ref{fig:2} is a graphical depiction of the process.

\begin{figure}[ht]
\begin{center}
\includegraphics[width=5in]{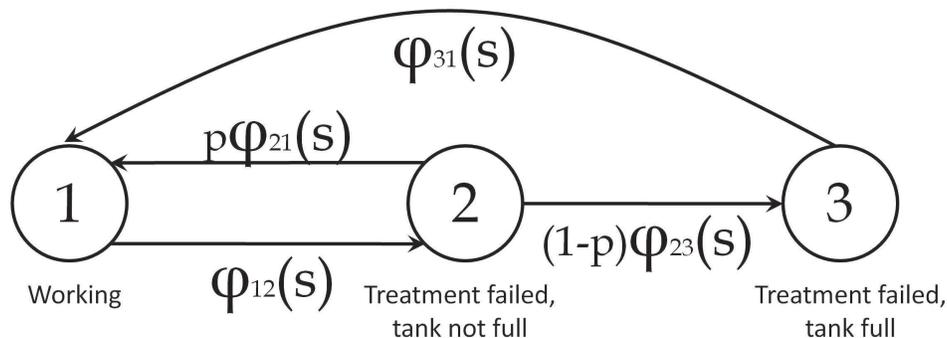}
\caption{A graphical depiction of the textile waste treatment process.}
\label{fig:2}
\end{center}
\end{figure}

We use the model parameterization of \cite{barbu}, which is:
\begin{eqnarray}
T_{12} &\sim& \text{geometric}(0.8),      			\nonumber \\
T_{21} &\sim& \text{discrete Weibull}(0.3,0.5),      \nonumber \\
T_{23} &\sim& \text{discrete Weibull}(0.5,0.7), \text{ and}      \nonumber \\
T_{31} &\sim& \text{discrete Weibull}(0.6,0.9).      \nonumber 
\end{eqnarray}

The discrete Weibull as defined in \cite{DiscreteWeibull} (slightly modified) is
\[ f(t|q,b) = \begin{cases} q^{(t-1)^b}-q^{t^b} &\mbox{if } t \in \{1,2, \ldots \} \\ 
                            0                   &\text{otherwise.} \end{cases} \]
The SMP parameter $p$ (see Figure \ref{fig:2}) is set at 0.95, or in words, on average 95\% of the repairs to the treatment facility are completed before the tank becomes full.

A few of the quantities that may be of interest from this process are the first passage PMF from state $1$ to state $3$, the expected amount of time spent in state $3$ at some time $t$, and the probability the factory is shutdown due to waste treatment failure at some time $t$.  The formulas for these quantities can be found in \cite{WarrCollins1}.  For this paper we focus on finding the first passage PMF from state $1$ to state $3$.

The first passage CF is given in Equation \ref{eq:ftpas}, where we replace CFs with FTs in all instances.  So to calculate the first passage FT from state 1 to state 3, we need the FT of $T_{12}$, $T_{21}$, and $T_{23}$.  The FT of geometric RV $T_{12}$ is known in closed form, however, the FT for the discrete Weibull is not.  So in our computations we use the exact FT for $T_{12}$ and use the DFT approximation for the other two.

We first want to calculate the mean and variance of $T_{12}$, $T_{21}$, and $T_{23}$ which will allow us to use the Cantelli inequality to bound our error.  Table \ref{tab1} shows the expectations and variances for these RVs.
\begin{table}
  \centering
    \caption{The mean and variance of waiting time random variables in the textile waste treatment process.}
\begin{tabular}{ ccc } 
\hline\noalign{\smallskip}
    & Mean & Variance \\   
    \noalign{\smallskip}\hline\noalign{\smallskip}
  $T_{12}$ & 1.2500 & 0.3125 \\
  $T_{21}$ & 2.0769 & 9.0765 \\
  $T_{23}$ & 2.7300 & 9.4618 \\ \hline
  $T_{13*}$ & 67.191 & 4394.1 \\ 
  \noalign{\smallskip}\hline
\end{tabular}

  \label{tab1}
\end{table}
These quantities allow us to calculate the number $N$ we need to use to obtain a desired error bound.  We choose $\varepsilon = 10^{-6}$ which implies $k \geq 66356$.  Therefore if $N=2^{17}$ the pointwise error for each probability will be smaller than 2$\varepsilon$.  The computation time of calculating the first passage PMF is less than one second (on a Windows 7 PC with an AMD Athlon$\rm {}^{TM}$ 7750 2.7 GHz dual-core processor).

The plot in Figure \ref{fig:3} shows what the first passage PMF from state $1$ to state $3$ looks like for the first one hundred points.  Using the developed theory in this paper, we are assured the pointwise error is less than or equal to 2$\varepsilon$.  Again this ignores the error introduced by multiplying and dividing the FTs (as in Equation \ref{eq:ftpas}) and the machine round-off error.

\begin{figure}[ht]
\begin{center}
\includegraphics[width=5in]{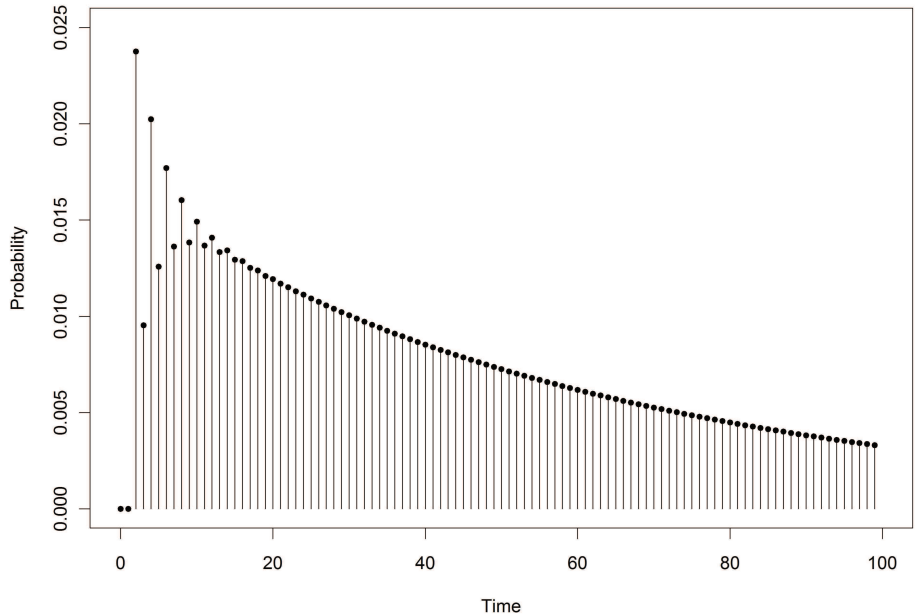}
\caption{The PMF of the first passage from state 1, ``Working", to state 3, ``Treatment failed, tank full".}
\label{fig:3}
\end{center}
\end{figure}

Figure \ref{fig:3} indicates that the PMF of this first passage distribution is monotonically decreasing after roughly time $20$.  Therefore we can attempt to apply Theorem \ref{thm2}.  Since we estimated that $P(T_{13*}=609)<\varepsilon$ Therefore if we choose $N=1218$ we only need to calculate roughly $1\%$ of the values in the FFT and can still claim the pointwise error is less than or equal to 2$\varepsilon$.  This example shows the potential computation savings of using Theorem \ref{thm2} for first passage random variables.  Also, Theorem \ref{thm2} does not assume a finite mean or variance, but a monotonically decreasing PMF after a selected point on the support.  In this example we worked our way, by trial and error, from a large $N$ down to a smaller $N$, however, we are also able to work our way up from a small $N$ to a larger $N$.  This can reduce the computation time to find the desired error bound.

\section{Discussion}
It has been our goal to show for lattice random variables the DFT/FFT computes the desired first passage PMF extremely quickly and to a prescribed accuracy.  It is often very desirable to handle probability problems in the transformed CF domain, but computational methods to invert them back into the time or spatial domain have been the largest deterrent.  The DFT is a excellent solution for the inversion of characteristic functions to PMFs.  The DFT is very fast using the FFT algorithm and it also provides error bounds for estimates. 

\section{Acknowledgments}

The authors would like to thank Patrick Chapin, David Collins and Aparna Huzurbazar for their advice, assistance, and encouragement.

\section{Appendix -- R Code}

\begin{verbatim}
##Discrete Weibull PMF
dweibulldisc<-function(x,q,b){
  temp1=q^((x-1)^b)-q^(x^b); temp2=x-floor(x)
  temp3=temp2<1e-15 ;  temp4<-x==0; temp1[temp4]<-0
  return(temp1*temp3)
}

##p is the probability of the treatment facility will be fixed 
## before the holding tank becomes full
p = 0.95

N = 2^17
support = 0:(N-1)

## The DFT samples of the two of the transition distributions
T21 = dweibulldisc(support,0.3,0.5)
T23 = dweibulldisc(support,0.5,0.7)

## The Fourier transforms of the 3 transition distributions
F12 = p*exp(-2*pi*1i*support/N)/(1-(1-p)*exp(-2*pi*1i*support/N))
F21 = fft(T21);  F23 = fft(T23)

## The Fourier transform of the first passage distribution from
## state 1 to state 3
F_1st_Pas_13 = ((1-p)*F12*F23)/(1-p*F12*F21)

## The first passage PMF from state 1 to state 3
T_1st_Pas_13 = 1/N*Re( fft(F_1st_Pas_13,inverse=T) ) 
\end{verbatim}

\end{document}